%% file: main.tex
\documentclass[a4paper,journal]{IEEEtran}
\IEEEoverridecommandlockouts
\newcommand\hmmax{0}
\newcommand\bmmax{0}

\usepackage{graphicx}
\usepackage{tikz}
\usepackage{pgfplots}
\usepackage{xcolor}
\usepackage{color, colortbl}
\usepackage{amsmath}
\usepackage{bbm}
\usepackage{amsfonts, amssymb}
\usepackage{mathtools}
\usepackage{cite}
\usepackage[nolist]{acronym}
\usepackage{amsmath}

\usepackage{booktabs} 		% cute tables
\usepackage{diagbox}		% diagonal cell division in table
\usepackage{upgreek}
\usepackage{bbm}
\usepackage{Files/bm}
\usepackage{Files/winsnotation}
\usepackage{Files/Symbols_OPL}
\usepackage{Files/LVcolours} % cute colours

\usepackage[ruled,vlined,noresetcount]{algorithm2e}
\usepackage{setspace}	 	% stretch on algorithm2e
\usepackage{comment}
\usepackage{physics}

\usepackage{array}
\newcolumntype{C}[1]{>{\centering\arraybackslash}p{#1}}

\makeatletter
\def\endthebibliography{%
  \def\@noitemerr{\@latex@warning{Empty `thebibliography' environment}}%
  \endlist
}
\makeatother
\usepackage{amsthm}
\theoremstyle{definition}

\usepackage{algpseudocode,amsmath}
\newtheorem{lemma}{Lemma}

\newcommand{\px}{p_{\mathrm{X}}}
\newcommand{\py}{p_{\mathrm{Y}}}
\newcommand{\pz}{p_{\mathrm{Z}}}

\newcommand{\PauliX}{\M{X}}
\newcommand{\PauliY}{\M{Y}}
\newcommand{\PauliZ}{\M{Z}}

\pgfplotsset{compat=1.17}
\begin{document}

%\title{Performance Analysis of Quantum \\ Topological Planar Code}
\title{Quantum codes for asymmetric channels: \\ ZZZY surface codes}
%\title{Quantum error-correcting codes for asymmetric channels: ZZZY surface codes}
%\title{Performance Analysis of Quantum \\ Error-Correcting Surface Codes}

%\author{%
%  \IEEEauthorblockN{tbd}
%  \IEEEauthorblockA{CNIT/WiLab, DEI, University of Bologna, Italy\\
%  	Email: \{ \}@unibo.it }
%}

%\author{%
%    \IEEEauthorblockN{Diego Forlivesi, Lorenzo Valentini, Marco Chiani \\}
%    \IEEEauthorblockA{CNIT/WiLab, DEI, University of Bologna, Italy \\
%  	Email: \{diego.forlivesi2, lorenzo.valentini13, marco.chiani\}@unibo.it }
%}

\author{Diego Forlivesi,~\IEEEmembership{Graduate~Student~Member,~IEEE,}
Lorenzo~Valentini,~\IEEEmembership{Member,~IEEE,}
        and~Marco~Chiani,~\IEEEmembership{Fellow,~IEEE}
\thanks{The authors are with the Department of Electrical, Electronic, and Information Engineering ``Guglielmo Marconi'' and CNIT/WiLab, University of Bologna, 40136 Bologna, Italy. E-mail: \{diego.forlivesi2,lorenzo.valentini13, marco.chiani\}@unibo.it.
Work funded in part by the European Union - Next Generation EU, PNRR project PRIN n. 2022JES5S2. 
%\thanks{Work supported in part by}
%Part of this work will be presented at IEEE ICC 2023, Rome, May 2023 ({\it arXiv:2302.13015}).
}
}

% Don't write page number 0 to the cover page.
\maketitle 
%\markboth{Submitted to IEEE }{}

\input{Files/Acronimi_SICMMA.tex}
\setcounter{page}{1}

\begin{abstract}
We introduce surface ZZZY codes, a novel family of quantum error-correcting codes designed for asymmetric channels. Derived from standard surface codes through tailored modification of  generators, ZZZY codes can be decoded by the \ac{MWPM} algorithm with a suitable pre-processing phase. The resulting decoder exploits the information provided by the modified generators without introducing additional complexity. 
ZZZY codes demonstrate a significant performance advantage over surface codes when increasing the channel asymmetry, while maintaining the same correction capability over depolarizing channel. 
\end{abstract}


%%%%%%%%%%%%%%%%%%%%%%%%%%%%%%%%%%%%%%%%%%%%%%%%%%%%%%
\section{Introduction}
%%%%%%%%%%%%%%%%%%%%%%%%%%%%%%%%%%%%%%%%%%%%%%%%%%%%%%

The construction of a quantum computer presents a significant hurdle due to the presence of errors, which can quickly undermine quantum information integrity if not managed effectively. Consequently, error correction is crucial for ensuring the reliability of quantum computation \cite{Sho:95, Got:09, Pfi:23, ZorDePGio:23}. 
Surface codes play a pivotal role in the architecture of first-generation quantum computers, owing to their high error thresholds, planar structure, locality, and availability of efficient decoders \cite{KriSebLac:22, BluDolEve:23, AchRajAle:22, ForValChi:24STM}. The most widely used decoder for these codes is the Minimum Weight Perfect Matching (MWPM) decoder \cite{Higg:22, Bro:23}. 
Quantum channels are often modeled as memoryless and depolarizing, meaning that the three Pauli errors $\M{X}$, $\M{Y}$, and $\M{Z}$ are equally likely to occur. However, asymmetries in the error event probabilities can be present in real quantum devices, often due to different relaxation and dephasing times \cite{Eva07:XZAsymmetry, SarKlaRot:09}. 
To protect information flowing through asymmetric channels, one can use ad-hoc asymmetric quantum codes or \ac{CSS} codes constructed from two classical codes with different error correction capabilities, such as surface codes with rectangular lattices \cite{SarKlaRot:09, ChiVal:20a, ForValChi24:JSAC}. Another possibility is to modify the surface code generators to gain some asymmetric error correction capabilities. 
An example of this approach are the XZZX surface codes, designed to address scenarios where qubit dephasing is the primary noise source \cite{AtaTucBar:21}. 
In these codes, each ancilla measures according to $\M{X}$ in the horizontal direction and $\M{Z}$ in the vertical direction, leading to generators with both $\M{X}$ and $\M{Z}$ operators.
%The XZZX decoder requires minimal changes compared to standard surface code decoder \cite{AtaTucBar:21}. 
%
For XZZX codes, the primal and dual lattices can be decoded independently, allowing the use of \ac{MWPM} decoding.
Another example of modifying the surface structure is described in \cite{Hig:23}, where all $\M{Z}$ generators are replaced with $\M{Y}$ generators. 
Due to these adjustments, the resulting codes are no longer CSS. Additionally,
belief-propagation is needed with MWPM to better exploit all information in circuit-level noise models
considered \cite{Hig:23}. 

In this letter, we propose new quantum codes, named ZZZY surface codes, in which a few Pauli $\M{Y}$ measurements are incorporated at carefully selected locations within the lattice. This approach aims to improve code performance over asymmetric channels while considering decoder complexity. 
To preserve the use of the \ac{MWPM} decoder, while admitting an additional low-complexity pre-processing phase, we choose to insert at most one $\M{Y}$ measurement per plaquette. The resulting ZZZY surface code shows a marked improvement in the correction of error patterns consisting of $\M{Z}$ operators.
Throughout the paper, we will delve into the details of the decoder, providing illustrative examples to support its description.

%Additionally, we present how to adjust the \ac{MWPM} decoder to be able to decode our class of codes in a two-stages fashion. 
%This decoder, maintains equivalent complexity and preserve the code distance, while effectively exploiting the information from $\M{Y}$ measurements. 
%Through the application of this two-stages \ac{MWPM}, ZZZY codes demonstrate error correction capabilities comparable to surface codes over a depolarizing channel. Moreover, with increasing channel asymmetry, ZZZY codes exhibit a substantial performance advantage over surface codes.

%%%%%%%%%%%%%%%%%%%%%%%%%%%%%%%%%%%%%%%%%%%%%%%%%%%%%%
\section{Preliminaries and Background}
\label{sec:preliminary}

%\subsection{Quantum Stabilizer Error-Correcting Codes}
%\label{subsec:QEC}

We indicate as $[[n,k,d]]$ a \ac{QECC} with a minimum distance of $d$, encoding $k$ logical qubits into a codeword of $n$ data qubits. 
Having a distance $d$ allows to correct all error patterns of weight up to $t = \lfloor(d-1)/2 \rfloor$.
The Pauli operators are denoted as $\PauliX, \PauliY$, and $\PauliZ$. 
Employing the stabilizer formalism, each code is characterized by $n-k$ independent and commuting operators $\M{G}_i \in \mathcal{G}_n$, termed stabilizer generators or simply generators, with $\mathcal{G}_n$ being the Pauli group on $n$ qubits \cite{Got:09}. The codewords are stabilized by the generators. %Then, the channel effect is described by operators acting on a codeword. 
The generators define measurements on quantum codewords without disturbing the original quantum state, obtained through the use of ancillary qubits. 
For instance, if we have a generator $\M{G}_i = \M{Y}_1\M{Z}_4\M{Z}_6$, it means that its associated ancilla qubit $A_i$ has to perform a $\M{Y}$ measurement on qubit 1 and $\M{Z}$  measurements on qubit 4 and 6.
Measuring the ancilla $A_i$, the output is 0 if the operator acting on the codeword state commutes with $\M{Y}_1\M{Z}_4\M{Z}_6$, and 1 if it anti-commutes. 
If any ancilla $A_i$ returns 1, the decoder detects the occurrence of an error operator, and intervenes to find an operator capable of correcting it, ultimately restoring a codeword state wherein all ancillas return 0. %If the error pattern has too many single qubit operators, the decoder could fail, producing a codeword which is valid but different from the original one.

%Suppose an error $\M{E}\in \mathcal{G}_n$ impacts a codeword, leading to the state $\M{E}\ket{\psi}$. A binary sequence $\V{s}$, also known as the error syndrome, can be extracted, where the $i$-th entry $s_i$ is zero if $\M{G}_i$ commutes with $\M{E}$, and $s_i=1$ if $\M{G}_i$ anticommutes with it. Ancillas measuring $s_i = 1$ will be referred to as switched on ancillas.

%We adopt the notation $[[n, k,d_\mathrm{X}/d_\mathrm{Z}]]$ for asymmetric codes able to correct all patterns up to $t_\mathrm{X} = \lfloor(d_\mathrm{X}-1)/2\rfloor$ Pauli $\M{X}$ errors and $t_\mathrm{Z} = \lfloor(d_\mathrm{Z}-1)/2\rfloor$ Pauli $\M{Z}$ errors. 
Among stabilizer codes we focus on surface codes. These have qubits arranged on a plane and require only local interaction between qubits \cite{BraKit:98, DenKitLan:02, wan:03, Rof:19, HorFowDev:12}.
Logical operators can be easily identified on surface codes: $\M{Z}_L$ ($\M{X}_L$) operator consists of a tensor product of $\M{Z}$'s ($\M{X}$'s) crossing horizontally (vertically) the lattice. 
An important feature of surface codes is that they can be decoded with the \ac{MWPM} algorithm \cite{Higg:22}. 
This decoder builds a graph where vertices correspond to error ancillas, and edges are weighted according to the number of qubits between them. 
Finally, matching these ancillas in pairs the \ac{MWPM} can localize the errors.

To analyze quantum codes, it is common to assume that errors occur independently and with the same statistics on the individual qubits of each codeword.
Moreover, qubit errors can manifest as Pauli $\M{X}$, $\M{Z}$, or $\M{Y}$, with probabilities $\px$, $\pz$, and $\py$, respectively. 
The overall probability of a generic qubit error is $p = \px + \pz + \py$.
Two possible models are the \emph{depolarizing channel}, where $\px = \pz = \py = p / 3$, and the \emph{phase-flip channel}, characterized by $p = \pz$ with $\px = \py = 0$.
We characterize an asymmetric channel by the asymmetry parameter $A = 2\pz /(p - \pz)$.
By the means of this parametrization, for $A=1$ we have the depolarizing channel, and for $A\to\infty$ we have the phase-flip channel.
In the case of a symmetric code we can approximate the logical error rate for $p \ll 1$ as \cite{ForValChi24:MacW}
\begin{align}
\label{eq:error_probWithBetaApprox}
p_\mathrm{L} 
&\simeq \left(1-\beta_{t+1}\right) \binom{n}{t+1}p^{t+1} \,.%+ (1-\beta_{t+2}) \binom{n}{t+2}p^{t+2}
\end{align}
where
\begin{align}
\label{eq:betaGen}
    \beta_j =1- \frac{1}{p^j}\sum_{i = 0}^{j}\binom{j}{i} \, \pz^i \, \sum_{\ell = 0}^{j-i} \binom{j-i}{\ell}\, \px^\ell \, \py^{j-i-\ell} f_j(i,\ell)\,
\end{align}
is the fraction of errors of weight $j$ that the decoder is able to correct, while $f_j(i,\ell)$ is the fraction of errors of weight $j$, with $i$ Pauli $\M{Z}$ and $\ell$ Pauli $\M{X}$ operators, which are not corrected.

%\subsection{Minimum Weight Perfect Matching}
%\label{subsec:MWPM}

%%%%%%%%%%%%%%%%%%%%%%%%%%%%%%%%%
\section{Quantum ZZZY Codes} %\label{sec:quantum_ZZZY}
\label{sec:logicals}

In this section we propose the ZZZY codes, belonging to the family of topological codes. These are obtained starting from the lattice of a non rotated surface code, by modifying some of the measurements of the generators, as shown in Fig.~\ref{Fig:sur_ZZZY}. We emphasize that these codes are still planar and they require only local connectivity between qubits. Moreover, for the decoding it is possible to employ the \ac{MWPM} algorithm, with the addition of some conditional statements.
 \begin{figure}[t]
 	\centering
    % \resizebox{0.4\textwidth}{!}{
 	  %  \input{Figures/XXZYStructure}
    % }
    \includegraphics[width = 0.4\textwidth]{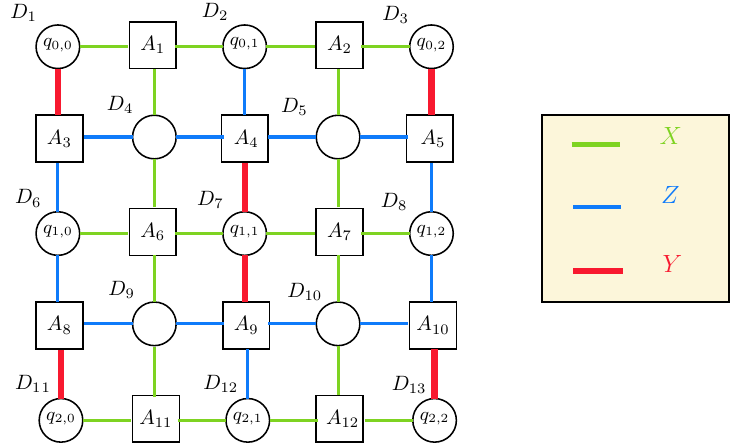}
 	\caption{ $[[13,1,3]]$ ZZZY code.  Circles stand for data qubits $D$, and squares for ancillae $A$. The six edges depicted in red denote a modified $\M{Y}$ measurement with respect the standard surface code. $\M{X}$,  $\M{Z}$, and  $\M{Y}$ measurements are depicted in green, blue, red, respectively. }
 	\label{Fig:sur_ZZZY}
 \end{figure}
 %

%\subsection{Structure of the Generators}

In the case of standard squared surface codes, each generator is responsible for only one kind of Pauli error (e.g., $\M{X}$ or $\M{Z}$), since they are composed by either all $\M{X}$ or all $\M{Z}$ operators.
As a result, such codes have a balanced error correction capability and perform best over symmetric channels. 
The basic idea behind ZZZY codes is to sacrifice some $\M{X}$  error correction capability to enhance the performance of the code over channels where phase flip errors are the most probable. 
Hence, we substitute a $\M{Z}$ with a $\M{Y}$ measurement for a subset of generators.
For instance, we design the $[[13,1,3]]$ ZZZY code with the following generators
 \begin{equation*}
 \arraycolsep=2.5pt
 \begin{array}{lllllll} 
 \M{G}_1 = \mathbf{X}_1 \mathbf{X}_2 \mathbf{X}_4 &
 \M{G}_2 = \mathbf{X}_2 \mathbf{X}_3 \mathbf{X}_5 & \\
 \M{G}_3 = \mathbf{Y}_1 \mathbf{Z}_4 \mathbf{Z}_6 &%\,, \quad
 \M{G}_4 = \mathbf{Z}_2 \mathbf{Z}_4 \mathbf{Z}_5 \mathbf{Y}_7 &%\,, \quad
 \M{G}_5 = \mathbf{Y}_3 \mathbf{Z}_5 \mathbf{Z}_8 \\
 \M{G}_6 = \mathbf{X}_4 \mathbf{X}_6 \mathbf{X}_7 \mathbf{X}_9 &%\,, \quad
 \M{G}_7 =  \mathbf{X}_5 \mathbf{X}_7 \mathbf{X}_8 \mathbf{X}_{10} & \\
\M{G}_8 = \mathbf{Z}_6 \mathbf{Z}_9 \mathbf{Y}_{11} &%\,, \quad
 \M{G}_9 = \mathbf{Y}_7 \mathbf{Z}_9 \mathbf{Z}_{10} \mathbf{Z}_{12} &%\,, \quad
 \M{G}_{10} = \mathbf{Z}_8 \mathbf{Z}_{10} \mathbf{Y}_{13} \\
 \M{G}_{11} = \mathbf{X}_9 \mathbf{X}_{11} \mathbf{X}_{12} &%\,, \quad
 \M{G}_{12} = \mathbf{X}_{10} \mathbf{X}_{12} \mathbf{X}_{13} \, &
 \end{array} 
 \end{equation*} 
which are shown in Fig~\ref{Fig:sur_ZZZY}. Hereafter, we will denote these modified generators as $\M{Z}\M{Y}$ generators.
 %Note that generators commutativity is ensured. 
 %Indeed, all generators either perform commuting measurements or share two anticommuting measurements. 
To build larger ZZZY codes, it is sufficient to start from the corresponding $[[n,k,d]]$ surface code, as follows. 
Considering data qubits only on odd rows, let assign two indices $i$ and $j$ to each data qubit in the lattice, where $i,j = 0, \dots, d-1$, denoting the row and column of the respective qubit $q$. Some examples of these labels are depicted in Fig.~\ref{Fig:sur_ZZZY}.
Next, transform the $\M{Z}$ measurements on qubits $q_{2\ell, 0}$ and $q_{2\ell, d-1}$, with $\ell = 0, \dots, d-1$, into $\M{Y}$ measurements.
Finally, convert the $\M{Z}$ measurements on qubits $q_{2\ell + 1, 1}$ and $q_{2\ell + 1, d-2}$ to $\M{Y}$ measurements. 
For the particular case $d=3$, depicted in Fig.~\ref{Fig:sur_ZZZY}, $d-2 = 1$ and then $q_{2\ell + 1, 1}$ and $q_{2\ell + 1, d-2}$ is the same qubit, for each $\ell$. 
This leads to $3(d - 1) = 6$ modifications to $\M{X}$ generators when $d=3$. It is easy to show that, when $d>3$, the procedure leads to $4(d - 1)$ modifications of $\M{X}$ generators.  
Note that the dual construction, where some $\M{X}$ are replaced by $\M{Y}$ to improve the error correction capability of bit flip errors, can be achieved in a similar manner.

%\subsection{Logical operators}
%\label{sec:logicals}

In the following, we will examine the logical operators of the $[[13,1,3]]$ ZZZY code to elucidate the advantage it attains in the presence of $\M{Z}$ channel errors.
The number of logical operators of each weight can be computed starting from Mac Williams identities as shown in \cite{ForValChi24:MacW}.
Specifically, for the $[[13,1,3]]$ %specified by the generators given above, 
using the approach in \cite{ForValChi24:MacW} we find that the undetectable error weight enumerator polynomial is 
\begin{align}
\label{eq:und_err}
&L(z)  =  6 z^3 + 24 z^4 + 75 z^5 + 240 z^6 + 648 z^7  + 1440 z^8 \notag \\ 
 &   + 2538z^9 + 3216z^{10} + 2634z^{11} + 1224z^{12} + 243z^{13} \, . 
\end{align}
Since this code has distance three, its asymptotic logical error rate depends on the fraction of errors of weight $j = 2$ that it is able to correct. 
In particular, it can be shown that Pauli errors of weight $j = 2$ can cause logical operators of weight $w = 3$ and $w=4$. From \eqref{eq:und_err}, we see that the $[[13,1,3]]$ ZZZY code has six logical operators with $w=3$ and 24 logical operators with $w=4$. 
Since this code is tailored for channels where phase flip errors occur more frequently, we focus on logical operators composed by only $\M{Z}$ Pauli operators. Referring to  Fig~\ref{Fig:sur_ZZZY}, some examples of logical operators with $w = 3$ and $w = 4$ are $\M{Z}_1$$\M{Z}_2\M{Z}_3$ and $\M{Z}_1$$\M{Z}_2\M{Z}_5\M{Z}_8$, respectively. 
The first one can be caused by three error patterns: $\M{Z}_1$$\M{Z}_2$, $\M{Z}_1$$\M{Z}_3$, and $\M{Z}_2$$\M{Z}_3$.
In the case of standard surface codes, where these errors are detected exploiting only information coming from $\M{X}$ generators, whenever one of these patterns occurs, the \ac{MWPM} is not able to recover it.
For instance, if the channel introduces a $\M{Z}_1$$\M{Z}_2$ error, the decoder will apply a $\M{Z}_3$, realizing the correspondent logical operator. 
However, ZZZY codes have additional information coming from $\M{Z}$$\M{Y}$ generators. 
Indeed, in case of a $\M{Z}_1$$\M{Z}_2$ occurs, ancilla qubit $A_3$, which performs $\M{Y}_1$$\M{Z}_4$$\M{Z}_6$ measurements, anticommutes with the error and it is switched on during the error correction. 
A similar reasoning can be done also for logical operators with $w = 4$. 
Specifically, these operators are due to $\binom{4}{2} - 2 = 4$ pattern of errors of weight two: $\M{Z}_1$$\M{Z}_5$, $\M{Z}_2$$\M{Z}_8$, $\M{Z}_1$$\M{Z}_8$, and $\M{Z}_2$$\M{Z}_5$. This is because  $\M{Z}_1$$\M{Z}_2$ causes a logical operator with $w = 3$, while $\M{Z}_5$$\M{Z}_8$ is always corrected.
In particular, the decoding error is due to the fact that the \ac{MWPM} is not able to distinguish between $\M{Z}_1$$\M{Z}_5$ and $\M{Z}_2$$\M{Z}_8$ ($\M{Z}_1$$\M{Z}_8$ and $\M{Z}_2$$\M{Z}_5$) since they give the same syndrome. However, a $\M{Z}_1$$\M{Z}_5$, contrary to $\M{Z}_2$$\M{Z}_8$, would switch on  $A_3$, which can be exploited to identify the correct channel error. 
The same can be said for each of the $\M{Z_L}$ logical operators with $w = 3, 4$. 
Hence, in the $[[13,1,3]]$ ZZZY code, all $\M{Z}$ error patterns of weight $t + 1$ are corrected, except for one: $\M{Z}_6$$\M{Z}_8$, resulting in $\beta_2 = 0.987$. 
This cannot be corrected as it results in the same syndrome as the error $\M{Y}_7$. 
%Note that for a code with $d = 3$ and the same parameters, this is the best achievable result.

%%%%%%%%%%%%%%%%%%%%%%%%%%%%%%%%%
\section{ZZZY Minimum Weight Perfect Matching }\label{sec:NumRes}

\begin{algorithm}[t]
\small
%\SetAlgoLined
\SetKwInOut{Input}{input}
\SetKwInOut{Output}{output}
\caption{\texttt{ZZZY\_Decoder}}\label{algo:MWPM}
\Input{%$\bm{e}$, vector of the channel errors; \\  
$\V{s}$, syndrome \\ $H$, matrix of the generators \\ $n_\mathbf{ZY}$, $n_\mathbf{X}$, number of $\M{Z}$$\M{Y}$ and $\M{X}$generators} 
\Output{$\bm{\hat{e}}$, vector of the estimated channel errors }
\BlankLine
%init $\hat{\V{e}}$ to all zeros \\
%$\bm{s} \gets \{ 0_1,...,0_{n-1}$\}, syndrome \\
init $\V{q}$ to all ones, vector of the weights associated to each data qubit of the lattice \\
$\V{q} \gets \texttt{update\_weights}(\V{s}, \V{q},  H, n_\mathbf{ZY}, n_\mathbf{X})$ \\
$D \gets \texttt{compute\_distance}(\V{s})$, matrix of the distances between switched on ancilla \\
$\hat{\V{e}} \gets \text{MWPM}_\mathbf{X}(D)$ \\
\ForAll{$i \in \{1,\dots,n_\mathbf{ZY}\}$}{
\ForAll{$j \in \{1,\dots,n\}$}{
\If{$\hat{\V{e}}(j) = 1$}{
   % \If{ $\bm{s}(i)  = 1$}{
     \If{$H(i,j) = 1$ \text{and} $H(i,j+n) = 1$}
    {
     $\V{s}(j) \gets 1 - \V{s}(j)$ 
    }
%}
}
}
}
$\hat{\V{e}} \gets \text{MWPM}_\mathbf{Z}(D)$ 
%\Return $\hat{\V{e}}$
\end{algorithm} 

\begin{algorithm}[t]
\small
\SetKwInOut{Input}{input}
\SetKwInOut{Output}{output}
\caption{$\texttt{update\_weights}$}%\label{algo:MWPM}
%\vspace{0.3cm}
%\hrule
%\vspace{0.01cm}
\Input{$\V{s}, \V{q},  H, n_\mathbf{ZY}$, $n_\mathbf{X}$} 
\Output{$\V{q}$}
\BlankLine
%\Function{$Update\_weights$}{$\bm{s}, \bm{q},  \mathsf{G}, n_\mathbf{Z}$} \\
%\SetAlgoNoLine
\label{algo:update_weights}
%\vspace{0.1cm}
%\hrule
%\vspace{0.05cm}
\ForAll{$i \in \{1,\dots,n_\mathbf{ZY}\}$}{
 \ForAll{$j \in \{1, \dots, n \}$ }{  
 \If{ $\bm{s}(i) = 1$}{
    \If{$H(i,j) = 1$ \text{and} $H(i,j+n) = 1$}
    {$\V{q}(j) \gets 0.9$
    }
    }
     \If{ $\V{s}(i) = 0$}{
    \If{$H(i,j) = 1$ \text{and} $H(i,j+n) = 1$}
    {$\V{q}(j) \gets 1.1$
    }
    }
}
}
%\ForAll{$i \in \{1,\dots,n_\mathbf{ZY}\}$}{
 %\ForAll{$j \in \{1, \dots, n \}$ }{  
 %\If{ $\V{s}(i) = 0$}{
   % \If{$H(i,j) = 1$ \text{and} %$H(i,j+n) = 1$}
    %{$\V{q}(j) \gets 1.1$
  %  }
 %   }
%}
%}
%\ForAll{$k \in \{1,\dots,n_\mathbf{X}\}$}{
%\If{$ \forall l \in \{1,\dots,n_\mathbf{X}\}$  s.t. $\mod(k, d_\mathbf{Z}) \neq \mod(l, d_\mathbf{Z}) : s(l) = 0$}{
init $\mathcal{A}$ to the empty set\\
\ForAll{$i \in \{1,\dots,n_\mathbf{ZY}\}$}{
\If{ $\V{s}(i) = 1$}{
$\mathcal{A} \gets \mathcal{A} \cup n_i$ \\
%$\bm{s}(i) \in \mathcal{A}$ 
%$\ell \gets f(n_i)$ \\
%$\mathcal{T} \gets g(\ell)$ \\
%\ForAll{$j \in \mathcal{T}$}{
\ForAll{$j \in g(h(n_i))$}{
\If {$\V{s(j)} = 1$} {
$\mathcal{A} \gets \mathcal{A}\setminus n_i$ 
}
}
}
}
\ForAll{$i \in \mathcal{A}$}{
 \ForAll{$j \in \{1, \dots, n \}$ }{  
% \If{ $\bm{s}(i) = 1$}{
    \If{$H(i,j) = 1$ \text{and} $H(i,j+n) = 1$}
    {$\V{q}(j) \gets -0.1$
    }
 %   }
}
}
%}
%}
%\Return $\V{q}$
%\EndFunction
%\vspace{0.1cm}
%\hrule
\end{algorithm}

In decoding ZZZY codes, we must adapt the standard \ac{MWPM} algorithm to leverage the insights gained from $\M{Y}$ measurements. Notably, as surface codes fall under the category of \ac{CSS} codes, the decoding process for $\M{Z}$ generators operates independently from that for $\M{X}$ generators \cite{FowMarMar:12}.
Consequently, the \ac{MWPM} can be divided into two phases: \ac{MWPM}$_\mathbf{X}$, focusing solely on $\M{X}$ generators, followed by \ac{MWPM}$_\mathbf{Z}$ for the $\M{Z}$ stabilizers.
As detailed in Section~\ref{sec:logicals}, $\M{Z}\M{Y}$ generators offer insights into certain $\M{Z}$ errors. However, without careful handling, they can erroneously trigger $\M{X}$ error detections. Take, for example, Fig~\ref{Fig:sur_ZZZY}, where a $\M{Z}_1$ error activates ancillas $A_1$ and $A_3$. Neglecting to deactivate ancilla $A_3$ before \ac{MWPM}$_\mathbf{Z}$ would falsely attribute an additional $\M{X}_1$ error. 
To address this, we introduce a preprocessing step to both \ac{MWPM}$_\mathbf{X}$ and \ac{MWPM}$_\mathbf{Z}$. 
%To maintain the algorithm's standard complexity, this preprocessing step should be low-complexity. 
The algorithm's complete description utilizes binary representation for the generators (i.e., for the parity check matrix $H$) and the estimated channel error vector $\hat{\V{e}}$. For instance, in a code with $n$ qubits, the matrix $H$ comprises $2n$ columns, with each row representing a generator. The first $n$ columns contain a 1 where the corresponding generator features a $\M{Z}$ or $\M{Y}$ Pauli measurement, while the second $n$ columns contain a 1 if the generators measure $\M{X}$ or $\M{Y}$ \cite{Got:09}.
We also use the first $n_\mathbf{ZY}$ rows to describe the $\M{Z}\M{Y}$ generators. 
The decoder for ZZZY codes is presented as Algorithm~\ref{algo:MWPM} above. Excluding the function \texttt{update\_weights} (to be introduced later), the algorithm ensures the minimum distance for ZZZY surface codes. 
After evaluating the syndrome, the function \texttt{compute\_distance} utilizes Dijkstra's algorithm to find the shortest paths on a graph, where vertices correspond to switched on ancillas and edges' weights are the sums of the underlying qubit weights.
Subsequently, via \ac{MWPM}$_\mathbf{X}$, pairs of $\M{X}$ ancillas are connected, producing the estimated $\M{Z}$ channel errors.
Next, the parity of all ancillas measuring $\M{Y}$ operators on qubits involved in $\M{Z}$ errors is inverted.
Finally, \ac{MWPM}$_\mathbf{Z}$ also allows for finding the $\M{X}$ channel errors. The algorithm corrects, therefore, all patterns of weight up to $t$. 
Further, we would like to correct as much as possible $\M{Z}$ errors of weight $t+1$. To this aim we can 
exploit the information coming from $\M{Z}\M{Y}$ generators. Specifically, after evaluating the syndrome, if some of the $\M{Z}\M{Y}$ generators are activated, we modify the weights of the edges of the \ac{MWPM} graph using the function \texttt{update\_weights}. 
Since we are considering minimum weight decoders, error patterns of weight $t+1$ could trigger only logical operators of weight $d$ and $d+1$ (if, as assumed, $d$ is odd). 
Let us start improving the correction in case of possible logical operators of weight $d+1$. 
We can achieve this if we apply the following procedure: if one of the generators performing a $\M{Y}$ measurement on the $i$-th qubit is activated, the weight $\bm{q}(i)$ of the corresponding edge is modified to a number slightly smaller than one, e.g., $\bm{q}(i) = 0.9$. 
Moreover, if a generator performing a $\M{Y}$ measurement on the $i$-th qubit is switched off, the weight $\bm{q}(i)$ is set to a number slightly larger than one, e.g., $\bm{q}(i) = 1.1$. 
In this way, during the \ac{MWPM}$_\mathbf{X}$, the decoder is pushed to choose paths where the $\M{Z}\M{Y}$ generators are switched on. 
If the $i$-th qubit is actually affected by a $\M{Z}$ Pauli error, this strategy allows the decoder to choose correctly between different paths composed by the same number of edges. 
We elucidate this with an example reported in Fig \ref{Fig:ZZZY_err}a. 
In particular, if $\M{Z}$ errors occur on data qubits $D_6$ and $D_3$, $\M{Z}\M{Y}$ ancilla $A_5$ is switched on. 
If we directly apply \ac{MWPM}$_\mathbf{X}$, the decoder has to choose between three error patterns of the same weight: $\M{Z}_3\M{Z}_6$, $\M{Z}_2\M{Z}_4$, and $\M{Z}_5\M{Z}_7$. This ambiguity could lead to an error with high probability. 
However, with our modification, the weight of qubit $D_3$ is set to $0.9$, guiding \ac{MWPM}$_\mathbf{X}$ to select it for correction.
Let us now focus on logical operators of weight $d$. In this case, if an error pattern with $t+1$ Pauli $\M{Z}$ operators occurs activating a $\M{Z}\M{Y}$ generator, we would like the decoder to select a path composed of a higher number of qubits if certain conditions are met. In doing so, we need to be sure that we are dealing with a potential logical operator of weight $d$. For this reason, if a $\M{Z}\M{Y}$ generator measuring a $\M{Y}$ operator on qubit $i$-th is activated, and there are no $\M{X}$ generators activated in the rows of the lattice adjacent to the one of qubit $i$-th, $\bm{q}(i)$ is set to a small negative number, e.g., $-0.1$, to force its selection. 
To formalize the algorithm, let us define $h(\cdot)$ as a function that takes as input the index of a $\M{Z}\M{Y}$ generator and returns the index $\ell$ of the qubit under $\M{Y}$ measurement. Additionally, we define the function $g(\cdot)$, which takes as input a qubit index and returns a list of $\M{X}$ generator indexes located in the row above and in the row below the input qubit. This function can be implemented efficiently using modulo operations. An example is depicted in Fig \ref{Fig:ZZZY_err}~b. Specifically, $\M{Z}$ errors have occurred on qubits $D_2$ and $D_3$. Applying the function $h(\cdot)$ to $A_5$, we obtain $h(5) = 3$, representing $D_3$. Consequently, $g(3)$ returns $\{ 6, 7\}$.The list has only two elements due to the fact that $D_3$ is on a boundary. Since ancillas $A_6$ and $A_7$ are both deactivated, the weight of qubit data $D_3$, measured by $A_5$, is set to $-0.1$, ensuring the correction of the error. On the other hand, without our \texttt{Update\_weights}, the \ac{MWPM}$_\mathbf{X}$ decoder would apply a $\M{Z}_1$ correction, leading to the logical operator $\M{Z}_1\M{Z}_2\M{Z}_3$.
In the worst case scenario, for each of the $4(d-1)$ $\M{Z}$$\M{Y}$ generators we could perform an assignment based on two conditional statements.
In practice, this can be easily implemented in hardware by means of simple logic gates, resulting in a pre-processing complexity of $O(1)$.

\begin{lemma}
\label{lem:minmatch}
Given an $[[n,k,d]]$ ZZZY code, for $d > 3$, the fraction of $\M{Z}$ errors of weight $t+1$ that cannot be corrected by the ZZZY decoder over a phase flip channel is $ d \binom{d - 2}{t + 1} / \binom{n}{t + 1}$
\end{lemma}
\begin{proof}
Over a phase-flip channel, all errors of weight $t+1$ that can cause logical operators of weight $2t+2$ are corrected.
Indeed, the decoder has to choose between two solutions composed of the same number of qubits.
Hence, by modifying the weight of the paths using the function \texttt{update\_weights}, the actual error pattern is always identified.
In case the $t+1$ errors occur on the same row of the lattice, they can cause a logical operator of weight $2t + 1$. 
To correct these errors, it is necessary that at least one of them occurs on a qubit measured by one of the two $\M{Z}\M{Y}$ generators, since the ZZZY decoder has to set the weight of the corresponding qubit to $-0.1$.
Hence, the uncorrected error patterns for each of the $d$ rows are $\binom{d - 2}{t + 1}$.
Finally, the total number of $\M{Z}$ error patterns of weight $t + 1$ is $\binom{n}{t + 1}$.
\end{proof} 
Note that, as the code distance increases, the fraction of errors of weight $t+1$ that cannot be corrected becomes smaller.

%We remark that, as surface codes, ZZZY codes are degenerate. 
%Hence, if the decoder fails to identify the correct Pauli operator, but the correction  together with the channel error is in the stabilizer, the codeword is still recovered. 
%For instance, considering the $[[13,1,3]]$ ZZZY code, the Pauli error $\M{Z}_4\M{Z}_6$ activates only ancilla $A_1$. Moreover, the same syndrome is obtained in the case of a $\M{Y}_1$ channel error. In particular, an \ac{MWPM} decoder always relates this syndrome to a $\M{Y}_1$ Pauli error. However, in the case of a $\M{Z}_4\M{Z}_6$ channel error, the resulting Pauli operator is the generator $\M{G}_3 = \mathbf{Y}_1 \mathbf{Z}_4 \mathbf{Z}_6$, which leaves the codeword unchanged.

%
 \begin{figure}[t]
 	\centering
 	% \resizebox{0.42\textwidth}{!}{
 	%    \input{Figures/XZZY_errors.txt}
  %    }
     \includegraphics[width = 0.4\textwidth]{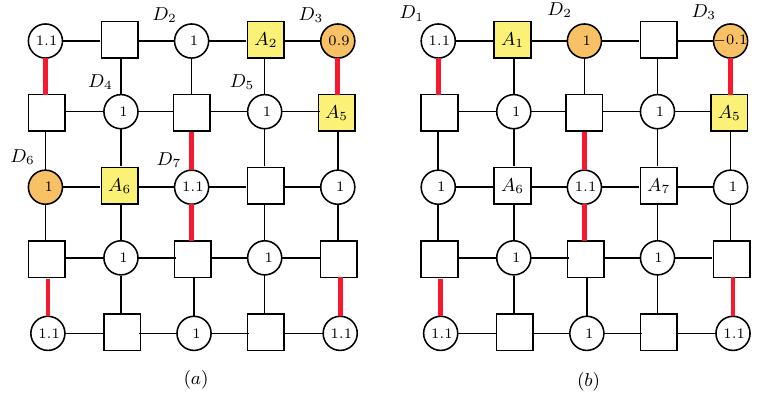}
 	\caption{ Decoding of the $[[13,1,3]]$ ZZZY code. Qubits affected by $\M{Z}$ errors are highlighted in orange. Switched on ancillas are depicted in yellow. Each qubit $i$ is associated with the corresponding weight $\V{q}(i)$ resulting from the function \texttt{Update\_weights}. }
 	\label{Fig:ZZZY_err}
 \end{figure}
 %
 
%%%%%%%%%%%%%%%%%%%%%%%%%%%%%%%%%%%%%%%%%%%%%%%%%%%%%%

\input{Figures/Tablee}

%%%%%%%%%%%%%%%%%%%%%%%%%%%%%%%%%
\section{Numerical Results}\label{sec:NumRes}

In this section we numerically evaluate the performance of ZZZY codes with the proposed decoder, providing a comparison with surface and XZZX codes under \ac{MWPM} decoding.  In Tab.~\ref{tab:Err} we report the fraction of non-correctable errors for each error class $f_j(i,\ell)$, evaluated by exhaustive search. 
We observe that, for the ZZZY codes, the values of $f_2(2,0)$ (i.e., the $\M{Z}\M{Z}$ class) and $f_3(3,0)$  (i.e., the $\M{Z}\M{Z}\M{Z}$ class) are the lowest.% among all cases.
This shows that ZZZY codes have the best $\M{Z}$ error correction capability.
%The corresponding values for surface and rotated surface codes can be found in \cite{ForValChi:23}.
Exploiting these tabular values, together with \eqref{eq:error_probWithBetaApprox} and \eqref{eq:betaGen}, we can evaluate the code performance.
To this aim, in Fig.~\ref{Fig:plot_symm} we report the logical error rates of surface and ZZZY codes with $d = 3$ and $d =5$ for a physical error rate $p = 0.001$, varying the channel asymmetry.
We note that ZZZY codes, while exhibiting comparable error correction capabilities with respect to surface codes over a depolarizing channel ($A=1$), show a significant performance advantage as the channel's asymmetry increases ($A>1$). 
For $A<1$ we can just use the dual version of the ZZZY code, which will give the same performance as for $A>1$.
%This behaviour is confirmed more in general for all investigated codes. 
Finally, Fig.~\ref{Fig:plot_phy} shows, for $A=100$, a comparison among the codes when varying the physical error rate. 
We observe that for high physical error rate the advantage of ZZZY codes over XZZX codes diminishes.

\begin{figure}[t]
	\centering
	% \resizebox{0.82\columnwidth}{!}{ 
	%     \input{Figures/plot_symm} 
	% } 
    \includegraphics[width = 0.82\columnwidth]{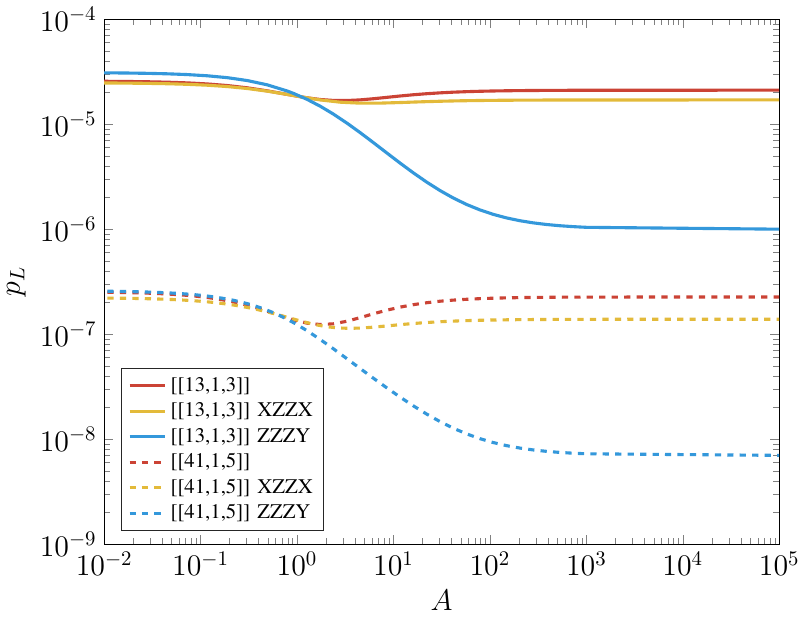}
	\caption{  Logical error probability vs. channel asymmetry for a physical error rate $p = 0.001$. Surface, XZZX, and ZZZY codes with $d = 3$ and $ d = 5$.}
		\label{Fig:plot_symm}
\end{figure}
\begin{figure}[t]
	\centering
	% \resizebox{0.82\columnwidth}{!}{ 
	%     \input{Figures/plot_vs_physical} 
	% } 
    \includegraphics[width = 0.82\columnwidth]{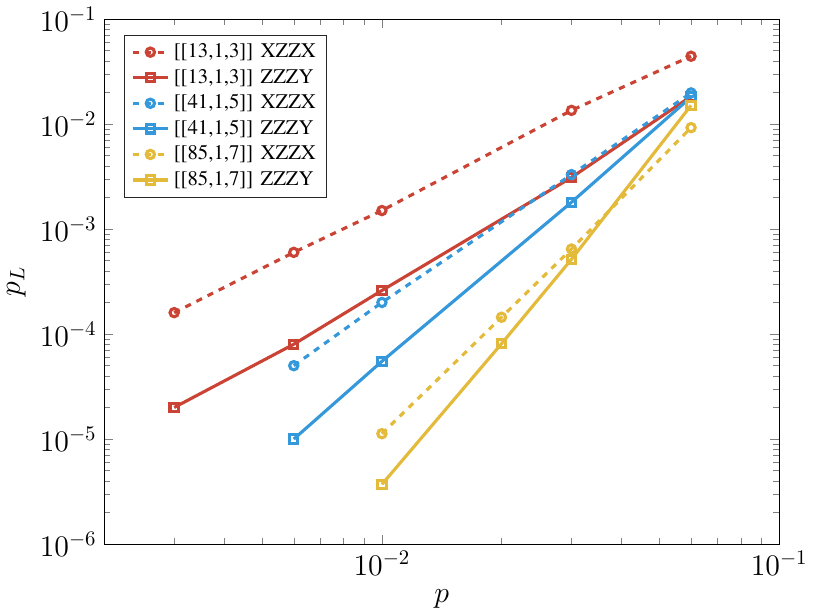}
	\caption{   Logical error probability $p_L$ vs. physical error rate $p$. XZZX, and ZZZY codes with $d = 3$, $ d = 5$, and $ d = 7$. Asymmetry $A = 100$.}	
 \label{Fig:plot_phy}
\end{figure}
% 

%%%%%%%%%%%%%%%%%%%%%%%%%%%%%%%%%%%%%%%%%%%%%%%%%%%%%%
\section{Conclusions}\label{sec:conclusions}
%%%%%%%%%%%%%%%%%%%%%%%%%%%%%%%%%%%%%%%%%%%%%%%%%%%%%%

We have introduced a novel family of \ac{QECC}, specifically designed for asymmetric channels and named ZZZY codes. 
These codes are derived from standard surface codes through the modification of certain generators. Furthermore, we have presented a variant of the \ac{MWPM} decoder, tailored for these codes. 
Remarkably, this decoder effectively leverages the augmented information from the modified generators without adding complexity. 
By employing our variant of the \ac{MWPM} decoder, the ZZZY codes 
exhibit a significant performance advantage compared to surface codes over asymmetric channels.

%%%%%%%%%%%%%%%%%%%%%%%%%%%%%%%%%%%%%%%%%%%%%%%%%%%%%%

%\section*{Acknowledgment}
%This work has been carried out in the framework of the CNIT National Laboratory WiLab.

\bibliographystyle{IEEEtran}
\bibliography{Files/IEEEabrv,Files/StringDefinitions,Files/StringDefinitions2,Files/refs}
%%%%%%%%%%%%%%%%%%%%%%%%%%%%%%%%%%%%%%%%%%%%%%%%%%%%%%

%\section{Appendix}\label{sec:Appendix}

\end{document}

%% file: Files/Acronimi_SICMMA.tex
\begin{acronym}
% usage: \ac{SW}, \acp{SW} for plurals \acf{SW} Use the full name of the acronym.
%\acs{SW}Use the acronym, even before the first corresponding \ac command
%\acl{acronym}Expand the acronym without using the acronym itself.
\small
\acro{AWGN}{additive white Gaussian noise}
\acro{BCH}{Bose–Chaudhuri–Hocquenghem}
\acro{CDF}{cumulative distribution function}
\acro{CRC}{cyclic redundancy code}
\acro{LDPC}{low-density parity-check}
\acro{ML}{maximum likelihood}
\acro{MWPM}{minimum weight perfect matching}
\acro{QECC}{quantum error correcting code}
\acro{PDF}{probability density function}
\acro{PMF}{probability mass function}
\acro{MPS}{matrix product state}
\acro{WEP}{weight enumerator polynomial}
\acro{WE}{weight enumerator}
\acro{BD}{bounded distance}
\acro{QLDPC}{quantum low density parity check}
\acro{CSS}{Calderbank, Steane, and Shor}

\end{acronym}

%% file: Figures/Tablee.tex
\begin{table*}[t]
    %\justifying
    \centering
    \setlength{\tabcolsep}{1pt}
    \caption{Fraction of non-correctable error patterns $f_j(i,\ell)$.}
    \label{tab:Err}
    \footnotesize
    %\resizebox{0.99\textwidth}{!}{
    %\begin{tabular}{l  C{2.5cm} C{2.25cm} C{2.5cm} C{2.75cm}}
    \begin{tabular}{lC{1.3cm}C{1.3cm}C{1.3cm}C{1.3cm}C{1.3cm}C{1.3cm}C{1.3cm}C{1.3cm}C{1.3cm}C{1.3cm}}
        \toprule
        %\multirow{2}{*}{
     %   \rowcolor[gray]{.95}
      % & & & & & \textbf{Surface} & & & & &\\
     %   \midrule
        \rowcolor[gray]{.95}
        \textbf{Code} &  $\M{X}\M{X}$ & $\M{X}\M{Z}$ & $\M{X}\M{Y}$ & $\M{Z}\M{Z}$ & $\M{Z}\M{Y}$ & $\M{Y}\M{Y}$ & & & &\\
        \midrule
        $[[13,1,3]]$ & $0.27$ & $0$ & $0.27$  & $0.27$  & $0.27$ & $0.51$ \\
        $[[13,1,3]]$ XZZX & $0.22$ & $0.051$ & $0.27$  & $0.22$  & $0.27$ & $0.51$ \\
        $[[13,1,3]]$ ZZZY & $0.27$ & $0.013$ & $0.37$  & $0.013$  & $0.28$ & $0.59$ \\
        %$[[23,1,3/5]]$ & $0.154$ & $0$ & $0.154$ & $0$ & $0.036$ & $0.130$ \\
        \midrule 
        \rowcolor[gray]{.95}
        \textbf{Code} & $\M{X}\M{X}\M{X}$ & $\M{X}\M{X}\M{Z}$ & $\M{X}\M{X}\M{Y}$ & $\M{X}\M{Z}\M{Z}$ & $\M{X}\M{Z}\M{Y}$ & $\M{X}\M{Y}\M{Y}$ & $\M{Z}\M{Z}\M{Z}$ & $\M{Z}\M{Z}\M{Y}$ & $\M{Z}\M{Y}\M{Y}$ & $\M{Y}\M{Y}\M{Y}$\\
        \midrule
        %$[[23,1,3/5]]$ & $0.384$ & $0.154$ & $0.384$ & $0.022$ & $0.188$ & $0.364$ & $0.014$ & $0.106$ & $0.256$ & $0.392$ \\
        $[[41,1,5]]$  & $0.021$ & $0$ & $0.021$ & $0$ & $0$ & $0.021$ & $0.021$ & $0.021$ & $0.021$ & $0.042$ \\
        $[[41,1,5]]$ XZZX & $0.014$ & $0.002$ & $0.016$ & $0.002$ & $0.007$ & $0.021$ & $0.013$ & $0.016$ & $0.021$ & $0.042$ \\
        $[[41,1,5]]$ ZZZY & $0.021$ & $0$ & $0.021$ & $0.001$ & $0.005$ & $0.021$ & $5\cdot 10^{-4}$ & $0.008$ & $0.020$ & $0.047$ \\
\bottomrule

    \end{tabular}
    %} $6.05\cdot 10^{-5}$
\end{table*}

%% file: main.bbl
\begin{thebibliography}{10}
\providecommand{\url}[1]{#1}
\csname url@samestyle\endcsname
\providecommand{\newblock}{\relax}
\providecommand{\bibinfo}[2]{#2}
\providecommand{\BIBentrySTDinterwordspacing}{\spaceskip=0pt\relax}
\providecommand{\BIBentryALTinterwordstretchfactor}{4}
\providecommand{\BIBentryALTinterwordspacing}{\spaceskip=\fontdimen2\font plus
\BIBentryALTinterwordstretchfactor\fontdimen3\font minus
  \fontdimen4\font\relax}
\providecommand{\BIBforeignlanguage}[2]{{%
\expandafter\ifx\csname l@#1\endcsname\relax
\typeout{** WARNING: IEEEtran.bst: No hyphenation pattern has been}%
\typeout{** loaded for the language `#1'. Using the pattern for}%
\typeout{** the default language instead.}%
\else
\language=\csname l@#1\endcsname
\fi
#2}}
\providecommand{\BIBdecl}{\relax}
\BIBdecl

\bibitem{Sho:95}
P.~W. Shor, ``Scheme for reducing decoherence in quantum computer memory,''
  \emph{Phys. Rev. A}, vol.~52, pp. R2493--R2496, Oct 1995.

\bibitem{Got:09}
D.~Gottesman, ``An introduction to quantum error correction and fault-tolerant
  quantum computation,'' in \emph{Quantum information science and its
  contributions to mathematics, {PSAM}}, vol.~68, 2010, pp. 13--58.

\bibitem{Pfi:23}
H.~D. Pfister, C.~Piveteau, J.~M. Renes, and N.~Rengaswamy, ``Belief
  propagation for classical and quantum systems: Overview and recent results,''
  \emph{IEEE BITS the Information Theory Magazine}, 2023.

\bibitem{ZorDePGio:23}
F.~Zoratti, G.~De~Palma, B.~Kiani, Q.~Nguyen, M.~Marvian, S.~Lloyd, and
  V.~Giovannetti, ``Improving the speed of variational quantum algorithms for
  quantum error correction,'' \emph{Phys. Rev. A}, vol. 108, no.~2, 2023.

\bibitem{KriSebLac:22}
S.~Krinner, N.~Lacroix, A.~Remm, A.~Di~Paolo, E.~Genois, C.~Leroux,
  C.~Hellings, S.~Lazar, F.~Swiadek, J.~Herrmann \emph{et~al.}, ``Realizing
  repeated quantum error correction in a distance-three surface code,''
  \emph{Nature}, vol. 605, no. 7911, pp. 669--674, 2022.

\bibitem{BluDolEve:23}
D.~Bluvstein, S.~J. Evered, A.~A. Geim, S.~H. Li, H.~Zhou, T.~Manovitz,
  S.~Ebadi, M.~Cain \emph{et~al.}, ``Logical quantum processor based on
  reconfigurable atom arrays,'' \emph{Nature}, p. 58–65, 2023.

\bibitem{AchRajAle:22}
{Google Quantum AI}, ``Suppressing quantum errors by scaling a surface code
  logical qubit,'' \emph{Nature}, vol. 614, no. 7949, pp. 676--681, 2023.

\bibitem{ForValChi:24STM}
D.~Forlivesi, L.~Valentini, and M.~Chiani, ``Spanning tree matching decoder for
  quantum surface codes,'' \emph{{IEEE} Commun. Lett.}, 2024.

\bibitem{Higg:22}
O.~Higgott, ``Pymatching: A {P}ython package for decoding quantum codes with
  minimum-weight perfect matching,'' \emph{ACM Transactions on Quantum
  Computing}, vol.~3, no.~3, pp. 1--16, 2022.

\bibitem{Bro:23}
B.~J. Brown, ``Conservation laws and quantum error correction: towards a
  generalised matching decoder,'' \emph{IEEE BITS the Info. Th. Mag.}, 2023.

\bibitem{Eva07:XZAsymmetry}
Z.~W.~E. Evans, A.~M. Stephens, J.~H. Cole, and L.~C.~L. Hollenberg, ``Error
  correction optimisation in the presence of x/z asymmetry,'' 2007.

\bibitem{SarKlaRot:09}
P.~K. Sarvepalli, A.~Klappenecker, and M.~R{\"o}tteler, ``Asymmetric quantum
  codes: constructions, bounds and performance,'' \emph{Proceedings of the
  Royal Society A}, vol. 465, no. 2105, pp. 1645--1672, 2009.

\bibitem{ChiVal:20a}
M.~Chiani and L.~Valentini, ``Short codes for quantum channels with one
  prevalent {Pauli} error type,'' \emph{IEEE JSAIT.}, vol.~1, no.~2, 2020.

\bibitem{ForValChi24:JSAC}
D.~Forlivesi, L.~Valentini, and M.~Chiani, ``Logical error rates of {XZZX} and
  rotated quantum surface codes,'' \emph{IEEE JSAC}, 2024.

\bibitem{AtaTucBar:21}
J.~P. Ataides, D.~K. Tuckett, S.~D. Bartlett, S.~T. Flammia, and B.~Brown,
  ``The {XZZX} surface code,'' \emph{Nat. Commun.}, vol.~12, no.~1, apr 2021.

\bibitem{Hig:23}
O.~Higgott, T.~C. Bohdanowicz, A.~Kubica, S.~T. Flammia, and E.~T. Campbell,
  ``Improved decoding of circuit noise and fragile boundaries of tailored
  surface codes,'' \emph{Phys. Rev. X}, vol.~13, p. 031007, Jul 2023.

\bibitem{BraKit:98}
S.~B. Bravyi and A.~Y. Kitaev, ``Quantum codes on a lattice with boundary,''
  \emph{arXiv preprint quant-ph/9811052}, 1998.

\bibitem{DenKitLan:02}
E.~Dennis, A.~Kitaev, A.~Landahl, and J.~Preskill, ``Topological quantum
  memory,'' \emph{Journal of Mathematical Physics}, vol.~43, no.~9, sep 2002.

\bibitem{wan:03}
C.~Wang, J.~Harrington, and J.~Preskill, ``Confinement-higgs transition in a
  disordered gauge theory and the accuracy threshold for quantum memory,''
  \emph{Annals of Physics}, vol. 303, no.~1, pp. 31--58, 2003.

\bibitem{Rof:19}
J.~Roffe, ``Quantum error correction: an introductory guide,''
  \emph{Contemporary Physics}, vol.~60, no.~3, pp. 226--245, jul 2019.

\bibitem{HorFowDev:12}
C.~Horsman, A.~G. Fowler, S.~Devitt, and R.~V. Meter, ``Surface code quantum
  computing by lattice surgery,'' \emph{New Journal of Physics}, vol.~14,
  no.~12, p. 123011, dec 2012.

\bibitem{ForValChi24:MacW}
D.~Forlivesi, L.~Valentini, and M.~Chiani, ``Performance analysis of quantum
  {CSS} error-correcting codes via {MacWilliams} identities,'' \emph{arXiv
  preprint arXiv:2305.01301}, 2024.

\bibitem{FowMarMar:12}
A.~G. Fowler, M.~Mariantoni, J.~M. Martinis, and A.~N. Cleland, ``Surface
  codes: Towards practical large-scale quantum computation,'' \emph{Physical
  Review A}, vol.~86, no.~3, sep 2012.

\end{thebibliography}
